\documentclass[runningheads]{llncs}
\usepackage{graphicx}
\usepackage{amssymb}
\usepackage{booktabs}
\usepackage{cite}
\usepackage{algorithm}
\usepackage{algpseudocode}
\usepackage{multirow}
\usepackage{amsmath,array,booktabs}
\usepackage{array} 
\newcommand\mc[1]{\multicolumn{1}{c}{#1}} 
\usepackage{subcaption}
\usepackage[utf8]{inputenc}
\usepackage{mathtools}
\usepackage{makecell}
\usepackage{lipsum}
\allowdisplaybreaks
\begin{document}
\title{A Set Cover Mapping Heuristic for Demand-Robust Fleet Size Vehicle Routing Problem with Time Windows and Compatibility Constraints}
\titlerunning{A Set Cover Mapping Heuristic for Demand-Robust Fleet Size Problem}
%
\author{Jordan Makansi, Ketan Savla}
\authorrunning{J. Makansi et al.}
%
\institute{University of Southern California, Los Angeles, USA 
\email{makansi@usc.edu}\\
}
\maketitle              
\begin{abstract}
We study the demand-robust fleet size vehicle routing problem with time windows and compatibility constraints. Unlike traditional robust optimization, which considers uncertainty in the data, demand-robust optimization considers uncertainty in which constraints must be satisfied.  This paper is the first to solve a practical demand-robust optimization problem at large scale.  We present an MILP formulation and also propose a heuristic that maps the problem to set cover in polynomial time.  We show that under modest assumptions the relative difference in time complexity from a standard branch-and-bound algorithm to the proposed heuristic scales exponentially with the size of the problem.  We evaluate our heuristic using a simulation case study on the Solomon benchmark instances for a variety of practical problem sizes, and compare with Gurobi.  The empirical approximation ratio remains below $2.0$.
\end{abstract}
\section{Introduction}
Traditionally, the objective in vehicle routing problems is to minimize some variation of the total cost.  While minimizing the cost is practical in the short term, the number of agents (fleet size) is assumed to be given during the routing phase.  A lesser-studied but perhaps more practical problem for long-term planning is the problem of determining the fleet size and mix of agents to purchase by stakeholders, so that they are available to service any routing problems that may arise \cite{corberan2021arc}.  \\
\indent Yet, even solving the fleet size problem alone has practical limitations. Its solutions determine the fleet size for a fixed set of routing constraints.  However, in any practical routing problem, requirements change frequently:  time-windows change, connections in a network become unavailable, new ones appear, etc.  Motivated by this need to address uncertainty in routing requirements when determining a fleet size and mix, we study the Demand-Robust Fleet Size Problem with Time Windows and Compatibility Constraints (DRFSP) by considering a finite number of scenarios, and determining fleet size by minimizing the maximum fleet cost over all scenarios. DRFSP is NP-Hard since two special cases are NP-Hard: the fleet size and mix problem \cite{FSM-Golden-1984}, and the vehicle routing problem with compatibility constraints \cite{yu2018approximation}.
\indent In the DRFSP we are given a finite number of scenarios, each with a set of customers which must be serviced. For each scenario the set of customers to be serviced has an associated time window in which the service must begin.  There may be several types of agents and only certain agents may service certain customers for each scenario.  If agents of a particular type may service a particular customer, the agent type and customer are "compatible".  There are 2 stages of decision-making in this problem; agents may be purchased in the first stage, or in the second stage at an inflated cost, specific to the agent type and the scenario. Each agents' path must start and end at a distinguished depot denoted $0$. The objective is to determine the number and type of agents to purchase in each stage, for each scenario, while minimizing the maximum sum of first and second stage costs, over all scenarios.\\
\indent There are three aspects of the literature that relate to our proposed 3-phase algorithm: fleet-size vehicle routing problem with time windows, compatibility constraints, and demand-robust optimization.  We briefly review the relevant literature for each.  Because it is NP-Hard, several heuristics have been proposed to solve the fleet size and mix problem with time windows: \cite{desrochers1991new, gheysens1986new, liu1999fleet}.  Initially savings heuristics were proposed, inspired by those for the fleet size and mix vehicle routing problem, and later \cite{dell2007heuristic} improved upon these results using a multi-start algorithm. 
In contrast to time windows, compatibility constraints are comparatively less common in the literature.  Most solution methods for compatibility constraints either involve exact methods, e.g. branch-and-cut-and-price \cite{ghiani2005waste, bustos2022drayage}, or approximations \cite{yu2018approximation} which show the efficacy of mapping to covering problems.  In the 3-phase heuristic algorithm proposed in this paper, we build upon heuristics for the two problems just described, by combining elements from both:  An insertion heuristic is used for generating routes, and the routes form a set-cover mapping used for solving a class of 2-stage problems.\\
\indent Robust optimization is used for preparing for the worst-case scenario, making it appropriate for long-term planning such as fleet-sizing.  While the routing literature has seen some exposure to robust optimization, it has been limited to uncertainty in the data, such as travel times, or demand serviced e.g. \cite{shen2022robust, zhang2023robust}.  However, in long-term planning the topology and the required arcs may change, which corresponds to a change which constraints must be satisfied. 
Since it was first conceived, demand-robust optimization has remained primarily in the computer science community \cite{dhamdhere2005pay}, and has been applied only to canonical problems such as shortest path, vertex cover, facility location, \cite{el2021power}. However, it has not been applied to any practical optimization problems of scale. Although some have suggested algorithms that are suitable for large-scale optimization, such as column generation \cite{bouman2011recoverable}, they still are limited to solving simple problems such as shortest path.  To the authors' knowledge this work is the first to address this gap in the literature; We design and evaluate the performance of a heuristic algorithm to solve demand-robust optimization to a large-scale practical problem motivated by the fleet sizing problem with time windows and compatibility constraints.\\
\indent We make the following contributions: (i) we propose the first application of demand robust optimization to a practical problem at scale and present its MILP formulation (ii) we propose a heuristic which we show has an exponential relative difference in time complexity compared to branch-and-bound, (iii) we present computational results on Solomon's benchmark instances, for a variety of problem instances, empirically showing an approximation ratio less than 2.0.
\section{Problem Formulation}
We are given a finite number of scenarios $k=1...m$, and for each scenario a "timetable", denoted $G_k$.  The timetable with only entries for a particular agent type is $G^t_k \subseteq G_k$.  Each entry in the timetable contains a customer $i$ and a time window $[e_i, l_i]$ within which the service of that customer must start, and the set of types of agents $T_i$ which are allowed to service that customer.  Note that not all customers need appear in the timetable, and also note that a customer may appear more than once in a timetable, with a different time window or a different set of agent types $T_i$.  Each customer also has an associated service time  $s_i$, which is how long it takes to service.  An example is in Example \ref{ex:example_1} and will be used throughout this paper.
  There are 2 stages of decision-making in this problem.  We are given the cost per agent per type in the first stage $c^t$, and an inflation factor for the cost of agents in the second stage $c^t_k = \sigma_k c^t$, $\sigma_k>1$ for each scenario $k$.  In the first stage, the decision-maker is blind to which scenario will be realized, and purchases agents at price $c^t$ for type $t$.  In the second stage, the scenario $k$ is realized, and the decision-maker purchases the remaining agents at a cost of $\sigma_k c^t$ required to service all of the arcs in the timetable.  Each agents' path must start and end at a distinguished depot denoted vertex $0$.  The objective is to determine the number and type of agents to purchase in the first stage and second stage for each scenario,  to minimize the maximum cost required over all of the scenarios. \\
\begin{example} \label{ex:example_1}
Example with 2 scenarios, 2 types, and 9 customers.
\end{example}
\begin{table}[] 
\caption{Timetables and routes for two scenarios}
\label{tab:timetable_examples}
\begin{tabular}{@{}llllllll@{}}
\multicolumn{3}{c}{Scenario 0}\\
\toprule
i & $e_i$  & $l_i$ & $s_i$  & $T_i$ & $r^{A,0}_0$ & $r^{A,1}_0$ & $r^{B,0}_0$\\ \midrule
2 & 0.0  & 193.0   & 10.0 & A &  \checkmark &   &\\
3 & 673.0& 793.0   & 10.0 & A &             & \checkmark & \\
4 & 152.0& 272.0   & 10.0 & A &  \checkmark &  & \\
7 & 644.0& 764.0   & 10.0 & B &   &  &\checkmark  \\
8 & 73.0 & 193.0   & 10.0 & AB&  \checkmark & &\checkmark \\ \bottomrule
\end{tabular}
\qquad \qquad
\begin{tabular}{@{}lllllll@{}}
\multicolumn{3}{c}{Scenario 1}\\
\toprule
i & $e_i$  & $l_i$& $s_i$ & $T_i$ &  $r^{A,0}_1$ & $r^{B,0}_1$\\ \midrule
1 & 0.0   & 960.0 & 10.0 &A  & \checkmark&  \\
2 & 73.0  & 193.0 & 10.0 &B  &  & \checkmark\\
4 & 644.0 & 764.0 & 10.0 &B  &  & \checkmark\\
5 & 73.0  & 193.0 & 10.0 &AB & \checkmark & \checkmark \\
9 & 371.0 & 491.0 & 10.0 &AB & \checkmark & \checkmark\\ \bottomrule
\end{tabular}
\end{table}
\begin{remark} Without compatibility constraints, DRFSP has an obvious solution: Solve fleet size problem with time windows for each scenario independently, and purchase all agents in the first stage for the scenario requiring the most agents.  This solution is guaranteed to be feasible for all other scenarios since all the other scenarios require fewer agents, and optimal since agents are always cheaper in the first stage.  Compatibility constraints however may make this solution infeasible for other scenarios, hence demand-robust optimization is appropriate.
\end{remark}
The MILP formulation for this problem is shown in equation (\ref{eq:milp_demand_robust}). Indices $s,k,t,p$ are for stage, scenario, agent type, agent index, resp., and $i,j$ are for customers. Vectors $x,y,t$ are variables. $x^{s,k,t,p}_{i,j} \in \{0,1\}$ represents if customer $i$ was serviced immediately before customer $j$, $t^{s,k,t,p}_{i}$ is real-valued and represents the time when service to customer $i$ begins, $y^{t,p} \in \{0,1\}$ represents agent $p$ of type $t$ was purchased in the first stage.  $R^t_i=1$ if customer $i$ may be served by agent type $t$, 0 otherwise. $P$ is maximum supply of agents available for purchase. $s_i$ is the time it takes to service customer $i$. $W_k, \sigma_k$ are the set of customers that must be serviced, and the inflation factor resp., for scenario $k$. $M$ is a sufficiently-large constant to ensure logical constraints are consistent.\\
\begin{subequations} \label{eq:milp_demand_robust}
\begin{alignat}{2}
& \min_{x^{s,k,t,p}_{i,j}, t^{s,k,t,p}_{i}, y^{t,p}, z}   && z   \label{eq:milp_obj_demand_robust}
\\
& &  & \text{s.t. } \sum^T_{t=1} c^t \sum^P_{p=1}  (y^{t,p} + \sigma_k \sum_{j \in W_k} x^{1,k,t,p}_{0,j}) \le z , \,\, \forall k  \label{eq:min_max_demand_robust}
\\
& &  & \sum_t^T \sum_p^P \sum_{i \in \{0, W_k\}, i \neq j} ( x^{0,k,t,p}_{i,j} +  x^{1,k,t,p}_{i,j}) \ge 1, \,\, \forall j \in W_k, \,\, \forall k , \label{eq:cover_req_arcs_demand_robust}
\\
& &  &   \sum_{j \in \{W_k, N\}} x^{s,k,t,p}_{0,j} \le 1, \, \forall p, t, k, s \label{eq:leave_original_depot_once_demand_robust}
\\
& &  &   \sum_{i \in \{0, W_k\}, i \neq j} x^{s,k,t,p}_{i,j} \le R^t_{j,k} \,\,  \forall j \in W_k, \forall p, t, k, s  \label{eq:compatibility_first_stage_demand_robust}
\\
& &  &  \sum_{i \in \{0, W_k\}, i \neq j} x^{s,k,t,p}_{i,j}= \sum_{i \in \{W_k, N\}, i \neq j} x^{s,k,t,p}_{j,i}, \, \forall j \in \{W_k\} \,\forall p, t, k, s \label{eq:flow_conservation_demand_robust}
\\
& &  &   (t^{s,k,t,p}_{i} + s_i + d_{i,j,k}) \le t^{s,k,t,p}_{j} + M(1- x^{s,k,t,p}_{i,j}),  \,\,\, \forall i,j \in W_k,\, \forall p, t, k, s \label{eq:consistent_demand_robust}
\\
& &  &  (t^{s,k,t,p}_{i} - e_{i,k}) \geq 0 \,\, \forall i \in W_k,  \forall p, t, k, s \label{eq:time_window_a_demand_robust}
\\
& &  &  (t^{s,k,t,p}_{i} -  l_{i,k}) \leq 0 \,\, \forall i \in W_k, \forall p,t,k,s \label{eq:time_window_b_demand_robust}
\\
& &  &  x^{0,k,t,p}_{0,j} -   y^{t,p} \leq 0, \,\, \forall j \in W_k, \,\, \forall p , \forall t , \forall k  \label{eq:first_stage_purchases_demand_robust}
\\
& &  &  x^{s,k,t,p}_{i,j} =0 \,\, \forall i,j \in W_k, \, i=j, \,  \forall p , \forall t , \forall k ,\,\,s\in\{0,1\}  \label{eq:no_self_loop_demand_robust}
\\
& &  &  x^{s,k,t,p}_{i,0} =0 \,\, \forall i \in W_k, \, \,  \forall p , \forall t , \forall k ,\,\,s\in\{0,1\}  \label{eq:no_enter_original_depot_demand_robust}
\\
& &  &  x^{s,k,t,p}_{N,j} =0 \,\, \forall j \in W_k, \,\,  \forall p , \forall t , \forall k ,\,\,s\in\{0,1\}  \label{eq:no_leave_final_depot_demand_robust}
\\
& &  &  x^{s,k,t,p}_{i,j}, y^{t,p} \in \{0,1\}, \,\, \,\, \forall i,j \in W_k, \,  \forall p , \forall t , \forall k ,\,\,s\in\{0,1\}  
\\
& &  &  t^{s,k,t,p}_{i}, z \in \mathbb{R^+}, \,\, \forall i \in W_k, \, \forall p , \forall t , \forall k ,\,\,s\in\{0,1\}
\end{alignat}
\end{subequations}
Constraint (\ref{eq:min_max_demand_robust}) is for the min-max objective. Constraint (\ref{eq:cover_req_arcs_demand_robust}) ensures that exactly one agent visits each required customer. Constraint (\ref{eq:leave_original_depot_once_demand_robust}) ensures each agent leaves the original depot only once. Constraint (\ref{eq:compatibility_first_stage_demand_robust}) ensures compatibility. Constraint (\ref{eq:flow_conservation_demand_robust}) ensures flow conservation. Constraint (\ref{eq:consistent_demand_robust}) ensures consistency between service times and travel times (and also eliminates subtours in conjunction with (\ref{eq:leave_original_depot_once_demand_robust}) and (\ref{eq:flow_conservation_demand_robust})). Constraints (\ref{eq:time_window_a_demand_robust}) and (\ref{eq:time_window_b_demand_robust}) ensure arcs serviced within the time window. Constraint (\ref{eq:first_stage_purchases_demand_robust}) ensures consistency between first and second stage purchases, and the remaining constraints are to ensure each used agent leaves the original depot $0$ and arrives at the final depot $N$, and for domain constraints.  
\section{3-Phase Set Cover Mapping (SCM) Heuristic Algorithm}
The proposed solution method has 3 phases:First, we solve the Fleet Size Problem with Time Windows (FSPTW) separately for each scenario, using an insertion heuristic. In the second phase, the routes given as output from FSPTW are used to construct "route sets" using a greedy approach.  In the third phase these "route sets" become an input to solve a demand robust weighted set cover problem to obtain the final solution.  We refer to these three phases as \textit{Generate Routes}, \textit{Construct Route Sets}, and \textit{Solve Demand Robust Weighted Set Cover}, respectively.
\subsection{Generate Routes}
The first phase approximates a solution to the Fleet Size Problem with Time Windows, for each agent type $t$ and each scenario $k$, and outputs $S^t_k$, set of all routes for scenario $k$ for agent of type $t$.\\
\begin{definition}[route]
A \textit{route} $r^{t,p}_k$ for agent $p$ of type $t$ for scenario $k$ is a realization of $\{x^{t,p}_{i,j}\}_{i,j}$ and $\{t^{t,p}_i\}_i$ variables which satisfy the time windows and compatibility constraints for a subset of entries in $G^t_k$.
\end{definition}
\begin{problem} \textbf{Fleet Size Problem with Time Windows} \label{prob:FSPTW}\\
\textbf{Input}: $G^t_k$ \\
\textbf{Output}: Routes for each agent, $r^{t,p}_{k}$, which service all customers within the time window.\\
\textbf{Objective}: Minimize the number of agents (routes) $|S^t_k|$.
\end{problem}
\begin{example} \label{ex:example_2}
Continuing Example \ref{ex:example_1}, we generate routes for agents of type $A$ in scenario 0 shown in Table \ref{tab:timetable_examples}. The initial route is $r^{A,0}_0= \{0,2,0\}$. We evaluate inserting customer $i=3$ before customer $j=2$, given the distance matrix shown in Table \ref{tab:distances}.  Route $\{0,3,2,0\}$ is not feasible (calculations shown in Table \ref{tab:sample_calc}). Route $\{0,2,3,0\}$ is evaluated, and its objective is $0.5(1183-210)+0.5(120)$.  This is repeated for customers $i=4,7,8$, and the best insertion is made according to objective (\ref{eq:savings_objective}).
\end{example}
\begin{table}
\captionsetup{labelfont=bf,
              justification=raggedright,
              singlelinecheck=false}
\caption{Example of customer insertion heuristic}
\begin{subtable}[b]{1.48\textwidth}
\captionsetup{singlelinecheck=false,justification=justified}
\caption{Calculations for inserting customer i=3 into initial route $r^{A,0}_0= \{0,2,0\}$}
\label{tab:sample_calc}                                                                        
\begin{tabular}[b]{@{} ll >{$}l<{$} ll @{}}                                                   
\toprule
j & route  & t'_j - t_j  & $l_i - t_i$ & TW Feasible  \\ \midrule
2 & \{0,3,2,0\}   
  & \begin{aligned}[t]
                     t_j&=100 \\
                     t'_j&=\max \{500, 673\} + 10 + 200
                   \end{aligned} & $793-\max \{500, 673\}$ &  No, $t'_2 > l_2$  \\ 
0 & \{0,2,3,0\}
  & \begin{aligned}[t]
                     t_j&=210 \\
                     t'_j&=\max \{ \max \{100,0\} + 10 + 200, 673\} \\
                         & + 10 + 500\\
                   \end{aligned} & $793-\max \{310,673\}$ &  Yes \\ 
\bottomrule
\end{tabular}
\end{subtable}
\begin{subtable}[b]{1.48\textwidth}
\captionsetup{singlelinecheck=false,justification=justified}
\caption{Matrix of travel times $d_{i,j}$ for customers used for routing customer $i=3$}
\label{tab:distances}
$\begin{array}{*{6}{c|}}
\mc{} & \mc{0} & \mc{2} & \mc{3} \\ \cline{2-4}
0     & 0      & 100    & 500    \\ \cline{2-4}
2     & 100    & 0      & 200    \\ \cline{2-4}
\end{array}$
\end{subtable}
\end{table}
We approximately solve  \textbf{Fleet Size Problem with Time Windows} shown in \textit{Problem \ref{prob:FSPTW}} for each $t,k$ by an insertion heuristic, while simultaneously ensuring time-feasibility using approaches in \cite{solomon1987algorithms}. In what follows in this section, we use $j$ to denote the customer immediately after the proposed insertion location of customer $i$, $t_j$ and $ t'_j$ to denote the original service time, and the new service time after inserting customer $i$ along the route, resp.
The heuristic works as follows: For a given agent type $t$ for a single scenario $k$, we construct routes for entries $G^t_k$ starting with an initial feasible route with a single customer (in our implementation, we chose the customer with $\min l_i$).  We then iterate over compatible customers $i$, and all possible insertion locations.  For each possible customer/location, we compute an objective, shown in (\ref{eq:savings_objective}) that measures the time disruption to the current route.  Customers are inserted in this manner until no more customers may be inserted without violating the time window constraints, and a new route is generated.  This repeats until all customers for the current scenario that are compatible with the selected agent type have been inserted into a route.
The objective described is 
\begin{subequations} \label{eq:savings_objective}
\begin{alignat}{2}
& \min_{i}   && \min_{j} \qquad (1-\phi)(t'_j - t_j)+\phi(l_i-t_i)   \label{eq:something_else}
\end{alignat}
\end{subequations}
which is a weighted combination of the delay in the customer $j$ immediately before $i$ in current route, and how urgent it is to service the inserted customer $i$. $\phi \in [0,1]$ is an algorithmic parameter that controls the tradeoff between these two time window objectives.  Since service can only start within $[e_i, l_i]$, the service start time of a feasible insertion is the maximum of the arrival time, and $e_i$.  We apply this procedure to our running example in Example \ref{ex:example_2}. The entire \textit{Generate Routes} phase is shown in Algorithm \ref{alg:generate_routes}, where lines \ref{start} through \ref{end} are the insertion heuristic for solving the \textbf{Fleet Size Problem with Time Windows}.\\
\setlength{\floatsep}{1pt}
\begin{algorithm}
\caption{GenerateRoutes}\label{alg:generate_routes}
\begin{algorithmic}[1]
\Require $G_k,\,\forall k$ 
\Ensure $S^t_k,\,\, \forall t,k$ 
\For {$k \in 1....m$} 
  \For {$t \in 1....T$} 
    \State $\bar{G}^t_k \gets G^t_k$ \Comment Create copy of $G^t_k$
    \While {$\bar{G}^t_k \cap S^t_k  \neq \emptyset$} \label{start}
      \While {exists an entry $\{i, (e_i, l_i), T_i \}_i$ that may be inserted feasibly}
        \State Insert entry $\{i, (e_i, l_i), T_i \}_i$ into route $r^{t,p}_{k}$ in a location that minimizes objective \ref{eq:savings_objective} \label{insertion}
        \State $\bar{G}^t_k \gets \bar{G}^t_k \setminus \{i, (e_i, l_i), T_i \}_i$ \Comment Remove entry from $\bar{G}^t_k$
      \EndWhile
      \State $S^t_k \gets  S^t_k \cup  \{r^{t,p}_{k}\}$ 
      \State $r^{t,p}_{k} = \emptyset$ \Comment Start new route.
    \EndWhile \label{end}
  \EndFor
\EndFor
\end{algorithmic}
\end{algorithm}
\subsection{Construct Route Sets}
We introduce a greedy heuristic which constructs sets of routes for each agent type, by maximizing its coverage of the compatible customers for each scenario.  These \textit{route sets} become inputs to the demand robust weighted set cover problem.
\begin{definition}[route set]
A \textit{route set} $S^t_l$ (indexed by $l$) for agent of type $t$ is a set of $m$ routes $\{r^{t,p}_1,...r^{t,p}_k,...r^{t,p}_m\}$ one for each scenario $k$.
\end{definition}
We generate the route sets $S^t_l$ in a greedy fashion:  For each agent type, we add select routes that maximize coverage of uncovered elements of the current timetable, and add them to the current route set.  This process is repeated until for each entries in all timetables  $G_k$, there exists a route set containing it.  Continuing our 2-scenario example from table \ref{tab:timetable_examples}, three route sets would be created, in the following order: $S^A_0=\{r^{A,0}_0, r^{A,0}_1\},\, S^A_1=\{r^{A,1}_0, r^{A,0}_1\},\, S^B_0=\{r^{B,0}_0, r^{B,0}_1\}$. We use $S^t=\{S^t_l \}^{L^t}_{l}$ to denote the collection of "route sets" for agent type $t$, and $L^t=|S^t|$ is the cardinality of the collection of route sets for agent type $t$ ($L^t \le \max_k |G^t_k|$).
  This ensures that the resulting demand robust weighted set cover problem is well-posed. The procedure is described in Algorithm \ref{alg:construct_route_sets}.
\begin{algorithm}
\caption{ConstructRouteSets}\label{alg:construct_route_sets}
\begin{algorithmic}[1]
\Require $S^t_k \,\, \forall t,k$ 
\Ensure $S^t, \,\, \forall t$ 
\For {$t \in 1....T$}
  \State $S^t = \emptyset$ 
  \While {for some $k$, $\exists \{i, (e_i, l_i), T_i \} \in G_k$ for scenario $k$, s.t. $t \in T_i$ and $\{i, (e_i, l_i), T_i \} \notin S^t$}: 
    \State $S^t_l = \emptyset$ 
    \For {$k \in 1....m$}
        \State Pick $r^{t,p}_{k}$ that maximizes $|r^{t,p}_{k} \cap (G_k \setminus S^t_l)|$. 
        \State $S^t_l \gets S^t_l \cup \{r^{t,p}_{k}\}$ 
    \EndFor
    \State $S^t \gets S^t \cup S^t_l$ 
  \EndWhile
\EndFor
\end{algorithmic}
\end{algorithm}
\subsection{Demand-Robust Weighted Set Cover}
In the third phase, \textit{Demand-Robust Weighted Set Cover (DRWSC)}, the route sets become parameters in solving Demand-Robust Weighted Set Cover problem, which is described as follows: We are given a universe $\mathcal{U}$, and a finite number of scenarios $A_k \subseteq \mathcal{U}$.  We are also given sets $S$, which may be purchased at a cost of $c_0(S)$ during the first stage, or at a cost $\sigma_k c_0(S)$ in the second stage.  Each element $e \in \mathcal{U}$ must be covered either in the first stage, or in the second stage for all scenarios containing it.  The objective is to minimize the maximum sum of first and second stage costs over all scenarios.
In the MILP formulation below, $x^{0}_{S}\in \{0,1\}$ are variables indicating set $S$ was purchased in first stage, and $x^{k}_{S}\in \{0,1\}$ are variables indicating set $S$ was purchased in second stage for scenario $k$.\\
\begin{subequations} \label{eq:demand_robust_weighted_set_cover}
\begin{alignat}{2}
& \min_{x^{0}_{S}, x^k_{S}}   && z   \label{eq:obj_demand_robust_weighted_set_cover}
\\
& &  & \text{s.t. } \sum_{S} c_0(S) (x^{0}_{S} + \sigma_k x^k_{S}) \le z , \,\, \forall k  \label{eq:min_max_demand_robust_weighted_set_cover}
\\
& &  & \sum_{S:e\in S} x^{0}_{S} + \sum_{S:e\in S} x^k_{S} \ge 1 , \,\, \forall k, \,\, \forall e \in A_k  \label{eq:coverage_demand_robust_weighted_set_cover}
\\
& &  &  x^{0}_{S}, x^k_{S} \in \{0,1\}, \,\, \forall k, S \label{eq:domain}
\end{alignat}
\end{subequations}
Constraint (\ref{eq:coverage_demand_robust_weighted_set_cover}) ensures that all entries in all elements in all scenarios are covered either in first stage or second stage purchases, constraint (\ref{eq:obj_demand_robust_weighted_set_cover}) minimizes the maximum over all scenarios, and objective (\ref{eq:domain}) is the constraint on the domains of the decision variables.\\
We map to DRWSC above using the \textit{route sets} introduced earlier:  Each entry in the timetable $G_k$ becomes an element $e$ in the problem (\ref{eq:demand_robust_weighted_set_cover}). The union of the entries in all of the timetables $\cup_{k=1} G_k$ becomes the universe $\mathcal{U}$. The route sets $S^t_l$ span across all scenarios, and become the sets $S$ available for purchase. The elements in each scenario $A_k$ map 1:1 to the entries in timetable $G_k$.  The costs $c_0(S), \sigma_k c_0(S)$ for first and second stage, respectively are given by the cost of the agent type $t$ corresponding to the collecion of route sets $S^t_l$. Continuing our example from before, route sets $S^A_0=\{r^{A,0}_0, r^{A,0}_1\},\, S^A_1=\{r^{A,1}_0, r^{A,0}_1\},\,$ and $S^B_0=\{r^{B,0}_0, r^{B,0}_1\}$ become sets $S$ in problem (\ref{eq:demand_robust_weighted_set_cover}), and the costs $c(S^A_0), c(S^A_1), c(S^B_0)$, are the costs $c^A, c^B$, respectively.
\section{Time Complexity Analysis}
We refer to our 3-phased set cover mapping heuristic as SCM throughout our analysis, and we refer to applying Gurobi directly for solving problem (\ref{eq:milp_demand_robust}) as MIP.  Since by default Gurobi is using branch-and-bound (BnB), wherever Gurobi is applied, we assume a standard BnB algorithm \cite{gurobi}. In BnB a divide-and-conquer approach is used to solve integer programs: The feasible region is divided into subregions, each of which corresponds to a node in the BnB tree.  A subproblem is solved, and an upper bound on the objective value for that subproblem is obtained.  If it is worse compared to the current best solution, the subproblem is eliminated,  otherwise the feasible region is further divided, creating new subproblems.  This is repeated until all subproblems are solved.  The time complexity of this standard BnB algorithm is $O(Mb^d)$, where $M$ is an upper bound on the time it takes to explore a subproblem at a node in the BnB tree, and $b,d$ are the number of branches and depth of the BnB tree\cite{morrison2016branch}. We assume each subproblem is solved via linear relaxation.  Gurobi uses simplex method to solve the linear relaxation, which is exponential in worst case, but in practice it is frequently polynomial in the number of variables and constraints \cite{spielman2004smoothed}.  So, in our worst case analysis, we are only concerned with the size of the BnB tree, $O(b^d)$.  Let $W$ to denote $\max_k |G_k|$, and DRWSC to denote solving the Demand-Robust Weighted Set Cover problem in the third phase of SCM.
\begin{remark} SCM does not need $P$ as an input, but the fleet size obtained after solving is given as input to MIP. So, $P$ is sufficiently high such that both SCM and MIP can find the same solution.
\end{remark}
\begin{lemma}
The asymptotic time complexities of SCM and MIP are $O(W^3 T m + WT + 2^{WT(m+1)})$, and $O(2^{2mTPW^2 + TP})$, respectively.
\end{lemma}
\begin{proof}
The worst case time complexity of Algorithm \ref{alg:generate_routes} is $O(W^3 T m)$. In \textit{Generate Routes} all routes, all customers, and all locations are evaluated - leading to $O(W^3)$ for a fixed type and scenario \cite{campbell2004efficient}.  The complexity of Algorithm \ref{alg:construct_route_sets} is $O(W T)$.  After Algorithm \ref{alg:construct_route_sets} completes we have mapped the problem to demand-robust weighted set cover in polynomial time.
DRWSC shown in equation \ref{eq:demand_robust_weighted_set_cover} is solved by BnB.  The maximum number of sets as input to problem (\ref{eq:demand_robust_weighted_set_cover}) is $WT$, and there are $m+1$ variables for each set $S$, leading to $WT(m+1)$ binary variables. The time complexity of DRWSC is $O(W^3 T m + WT + 2^{WT(m+1)})$.
The number of binary variables in solving problem (\ref{eq:milp_demand_robust}) is $2mTPW^2 +TP$.  Since problem (\ref{eq:milp_demand_robust}) is directly given to Gurobi, the time complexity of MIP is $O(2^{2mTPW^2 +TP})$.
\end{proof}
\begin{lemma} The relative difference in computation time between MIP and DRWSC increases at a rate of $O(2^{2mTPW^2})$.
\end{lemma}
\begin{proof}
The ratio of $O(2^{2mTPW^2 +TP})$ and $O(W^3 T m + WT + 2^{WT(m+1)})$ is $O(2^{2mTPW^2 +TP})/O(2^{WT(m+1)}) = O(2^{2mTPW^2})$.
\end{proof}
This analysis is consistent with the relative difference in solution times seen in the computational results, as $W, T, P, m$ increase.
\begin{remark} In comparing performance of branch-and-bound algorithms, a linear relaxation of the subproblem is typically used in the analysis \cite{wolsey1980heuristic}, hence we believe this assumption to be modest. 
\end{remark}
\begin{remark}  The presence of big-M constraints \ref{eq:consistent_demand_robust} in problem (\ref{eq:milp_demand_robust}) leads to a weak linear relaxation.  While the strength of the linear relaxation may not affect its worst-case running time, it will have a large effect on the quality of the integer bounds.  Thus, more of the $2^{2mTPW^2 +TP}$ nodes will be explored in the BnB tree, which, in addition to computation time is a typical metric used for measuring time complexity in comparing two branch-and-bound algorithms \cite{morrison2016branch}. DRWSC does not suffer from these big-M constraints.  
\end{remark}
\section{Computational Results}
We evaluate the performance of our proposed algorithm using Solomon's benchmark instances \cite{solomon1987algorithms}, briefly described as: three types of instances $R,C, $and $RC$ for random, clustered, and mixed. $101$/$201$ indicates narrow/wide time windows resp. We modify an instance to suit our problem as follows: for a given number of customers, types, and scenarios (denoted  $N$-$T$-$m$), we randomly sample $N$ customers with replacement for each of the $m$ scenarios.  For each number of types, we assign a probability that a particular customer is compatible with an agent (except for the depot, which is compatible with all types).  This probability is used to generate the entries in all of the timetables $G_k$.  The probabilities for each number of types are $10\%,40\%,$ and $80\%$ for $T=2,4,8$, respectively.  We select problems of practical size for applications such as road freight transportation \cite{zak2011multiple}.  For all examples, $\sigma_k=2$, $\phi=0.5$, and for each Solomon instance and $N$-$T$-$m$ tuple, we construct 10 examples independently, and aggregate their statistics.
The MIP model requires $P$ for solving problem \ref{eq:milp_demand_robust} - the maximum number of agents for any type.  We set $P$ sufficiently large such that SCM can find a solution, and use the same $P$ for solving using Gurobi.  All computation was done on USC's CARC high performance computing cluster, using Python 3.9.11, Gurobi 10.0.0, using 40G of RAM, running at 2.6GHz. We set a time limit of 10 minutes for Gurobi. \\
\indent For a comprehensive comparison with Gurobi, we investigate the performance from two perspectives:  First, we leave Gurobi's default settings, and set a time limit of 10 minutes.  Second, we seek the Gurobi settings under which it returns a similar computation time as SCM, and we compare the results in terms of solution quality.  By setting MIPFocus and MIPGap to $1$, and $1.0$, resp., Gurobi emphasizes finding many feasible solutions quickly, rather than exploring a single branch to full depth.  Gurobi terminates when the primal-dual gap is below MIPGap.  In our computational experiments, increasing the MIPGap beyond $1.0$ caused Gurobi to terminate prematurely with the trivial lower bound $0$.\\
\indent In the tables below, aggregated over 10 examples, the examples when Gurobi either exceeded the time limit, or reported a trivial lower bound of $0$, were removed before computing the mean and spread since they would be considered outliers.  In the table columns, time\% is the average percent reduction in time from MIP to SCM, i.e. $\frac{\text{time}_{MIP}-\text{time}_{SCM}}{\text{time}_{MIP}}*100$.  $\bar{\alpha}_{lb}, \bar{\alpha}_{ub}$ are the average values for $\frac{{OBJ}_{SCM}}{{MIP}_{LB}}$ and $\frac{{OBJ}_{SCM}}{{MIP}_{UB}}$, where ${OBJ}_{SCM}$ is the objective found by SCM, and ${MIP}_{LB},{MIP}_{UB}$ are the lower and upper bounds given by MIP. time\%std is the standard deviation of time\%. $\sigma_{\alpha_{lb}}$ , $\sigma_{\alpha_{ub}}$ are the standard deviations of $\bar{\alpha}_{lb},$ and $\bar{\alpha}_{ub}$ resp. TO is the number of times out of the 10 examples that Gurobi exceeded the time limit, and lb zero is the number of times that Gurobi reported a trivial lower bound of $0$. $\bar{t}_{SCM}$ is the average computation time taken by SCM, and T is the number of types. Missing data indicates the value could not be computed due to the  number of timeouts by MIP.\\
\indent For the default Gurobi parameters, for all of the Solomon instances tried, we see both the average percent relative time difference, and the number of time outs TO, increase with both the problem size, and the higher levels of compatibility (tables \ref{tab:r101_default}, through \ref{tab:rc201_default}).
  For small problem sizes, although MIP finds a solution faster than SCM, the average time for SCM still remains on the order of a few seconds.  With the exception of instance $C101$, when $N=25$ for the cases when MIP does not timeout for all 10 examples, time\% approaches $100\%$.  This trend is amplified as $T$ increases.  It is also worth noting that examples with 35, and 50 customers were tried for all levels of compatibility, and MIP ran out of memory in all of those examples, while SCM still finished computing in less than 10 seconds.  These observations are consistent with the exponential time difference in our time complexity analysis.  Finally, for all instances with default Gurobi parameters, the empirical approximation ratio of SCM never exceeds 2.0.\\
\indent For the gurobi settings MIPFocus=1, and MIPGap=1.0 in tables \ref{tab:r101} through \ref{tab:rc201}, again with the exception of instance $C101$, SCM still finds a solution with a significantly lower upper bound (as indicated by  $\bar{\alpha}_{ub} <1$). Across all instances with settings MIPFocus=1, and MIPGap=1.0, $\bar{\alpha}_{lb}$ never exceeds $4.5$.  Interestingly when solving instances with wide time windows $R201$ and $RC201$, MIP seems to either timeout, or almost immediately find a worse solution than SCM.  So, even when MIP approximates the solution in a similar amount of time as SCM, empirically its approximation ratios are notably worse. 
\setlength{\floatsep}{0.1pt}
\begin{table}
\captionsetup{labelfont=bf,
              justification=raggedright,
              singlelinecheck=false}
\caption{R101, default Gurobi settings.}
\label{tab:r101_default}
\begin{tabular}{@{}llllllllllll@{}}
\toprule
   N  & m & time\% & $\bar{\alpha}_{lb}$ & $\bar{\alpha}_{ub}$ & $\bar{t}_{SCM}$ & time\%std & $\sigma_{\alpha_{lb}}$ & $\sigma_{\alpha_{ub}}$ & TO & lb zero & T \\ \midrule
  5  & 2 & -96.06       & 1.02        & 1.02        & 2.62          & 162.82          & 0.08       & 0.08       & 0       & 0      & 2 \\
  10 & 3 & -149.21      & 1.14        & 1.14        & 2.98          & 162.66          & 0.1        & 0.1        & 0       & 0      & 2 \\
  15 & 5 & 41.11        & 1.14        & 1.14        & 2.99          & 64.13           & 0.1        & 0.1        & 0       & 0      & 2 \\
  25 & 3 & 66.09        & 1.14        & 1.14        & 1.99          & 36.62           & 0.15       & 0.15       & 1       & 0      & 2 \\\midrule
  5  & 2 & -34.07       & 1.35        & 1.35        & 1.72          & 56.23           & 0.43       & 0.43       & 0       & 0      & 4 \\
  10 & 3 & -30.99       & 1.53        & 1.53        & 2.23          & 94.44           & 0.21       & 0.21       & 0       & 0      & 4 \\
  15 & 5 & 13.06        & 1.44        & 1.44        & 2.98          & 57.2            & 0.21       & 0.21       & 0       & 0      & 4 \\
  25 & 3 & 87.49        & 1.55        & 1.55        & 3.16          & 24.65           & 0.13       & 0.13       & 2       & 0      & 4 \\\midrule
  5  & 2 & 9.36         & 1.35        & 1.35        & 2.82          & 78.62           & 0.34       & 0.34       & 0       & 0      & 8 \\
  10 & 3 & 93.27        & 1.42        & 1.42        & 3.1           & 7.86            & 0.2        & 0.2        & 6       & 0      & 8 \\
  15 & 5 & 98.39        & 1.52        & 1.52        & 2.53          & 1.76            & 0.23       & 0.23       & 5       & 0      & 8 \\
  25 & 3 & 97.73        & 1.62        & 1.62        & 2.37          &                 &            &            & 9       & 0      & 8 \\ \bottomrule
\end{tabular}
\caption{R201, default Gurobi settings.}
\label{tab:r201_default}
\begin{tabular}{@{}llllllllllll@{}}
\toprule
 N  & m & time\% & $\bar{\alpha}_{lb}$ & $\bar{\alpha}_{ub}$ & $\bar{t}_{SCM}$ & time\%std & $\sigma_{\alpha_{lb}}$ & $\sigma_{\alpha_{ub}}$ & TO & lb zero & T \\ \midrule
 5  & 2 & -190.13      & 1.0         & 1.0         & 2.95          & 136.18          & 0.0        & 0.0        & 0       & 0      & 2 \\
 10 & 3 & -258.01      & 1.0         & 1.0         & 3.82          & 158.4           & 0.0        & 0.0        & 0       & 0      & 2 \\
 15 & 5 & -43.45       & 1.2         & 1.2         & 2.6           & 100.7           & 0.26       & 0.26       & 0       & 0      & 2 \\
 25 & 3 & 86.39        & 1.7         & 1.7         & 4.82          & 10.71           & 0.38       & 0.38       & 1       & 0      & 2 \\\midrule
 5  & 2 & -58.56       & 1.4         & 1.4         & 2.09          & 92.76           & 0.29       & 0.29       & 0       & 0      & 4 \\
 10 & 3 & -169.16      & 1.2         & 1.2         & 3.72          & 156.66          & 0.17       & 0.17       & 0       & 0      & 4 \\
 15 & 5 & -58.08       & 1.05        & 1.05        & 3.18          & 113.84          & 0.11       & 0.11       & 0       & 0      & 4 \\
 25 & 3 & 11.1         & 1.36        & 1.36        & 5.14          & 120.29          & 0.21       & 0.21       & 0       & 0      & 4 \\\midrule
 5  & 2 & -21.96       & 1.8         & 1.8         & 2.56          & 114.33          & 0.63       & 0.63       & 0       & 0      & 8 \\
 10 & 3 & 75.83        & 1.35        & 1.35        & 2.55          & 14.04           & 0.34       & 0.34       & 0       & 0      & 8 \\
 15 & 5 & 97.13        & 1.96        & 1.96        & 5.29          & 1.0             & 0.29       & 0.29       & 2       & 0      & 8 \\
 25 & 3 &              &             &             & 6.65          &                 &            &            & 10      & 2      & 8 \\ \bottomrule
\end{tabular}
\caption{C101, default Gurobi settings.}
\label{tab:c101_default}
\begin{tabular}{@{}lllllllllllll@{}}
\toprule
N  & m & time\% & $\bar{\alpha}_{lb}$ & $\bar{\alpha}_{ub}$ & $\bar{t}_{SCM}$ & time\%std & $\sigma_{\alpha_{lb}}$ & $\sigma_{\alpha_{ub}}$ & TO & lb zero & T \\ \midrule
5  & 2 & -368.02      & 1.02        & 1.02        & 5.14          & 238.4           & 0.08       & 0.08       & 0       & 0      & 2 \\
10 & 3 & -82.12       & 1.0         & 1.0         & 1.68          & 88.44           & 0.0        & 0.0        & 0       & 0      & 2 \\
15 & 5 & -61.81       & 0.92        & 0.92        & 2.39          & 140.86          & 0.25       & 0.25       & 0       & 0      & 2 \\
25 & 3 & -57.43       & 1.08        & 1.08        & 2.5           & 58.73           & 0.07       & 0.07       & 0       & 0      & 2 \\\midrule
5  & 2 & -119.33      & 1.4         & 1.4         & 2.39          & 184.9           & 0.36       & 0.36       & 0       & 0      & 4 \\
10 & 3 & -107.46      & 1.54        & 1.54        & 1.82          & 112.88          & 0.31       & 0.31       & 0       & 0      & 4 \\
15 & 5 & -54.3        & 1.49        & 1.49        & 2.43          & 80.88           & 0.19       & 0.19       & 0       & 0      & 4 \\
25 & 3 & -23.57       & 1.45        & 1.45        & 3.03          & 73.31           & 0.19       & 0.19       & 0       & 0      & 4 \\\midrule
5  & 2 & -103.23      & 1.22        & 1.22        & 1.79          & 57.02           & 0.19       & 0.19       & 0       & 0      & 8 \\
10 & 3 & -175.39      & 1.52        & 1.52        & 3.3           & 161.21          & 0.38       & 0.38       & 0       & 0      & 8 \\
15 & 5 & 48.33        & 1.35        & 1.35        & 2.31          & 18.1            & 0.31       & 0.31       & 1       & 0      & 8 \\
25 & 3 & 59.1         & 1.41        & 1.41        & 3.27          & 27.56           & 0.16       & 0.16       & 0       & 0      & 8 \\ \bottomrule
\end{tabular}
\end{table}
\begin{table}
\captionsetup{labelfont=bf,
              justification=raggedright,
              singlelinecheck=false}
\caption{RC101, default Gurobi settings.}
\label{tab:rc101_default}
\begin{tabular}{@{}llllllllllll@{}}
\toprule
N  & m & time\% & $\bar{\alpha}_{lb}$ & $\bar{\alpha}_{ub}$ & $\bar{t}_{SCM}$ & time\%std & $\sigma_{\alpha_{lb}}$ & $\sigma_{\alpha_{ub}}$ & TO & lb zero & T \\ \midrule
5  & 2 & -115.59      & 1.05        & 1.05        & 2.16          & 192.63          & 0.16       & 0.16       & 0       & 0      & 2 \\
10 & 3 & -75.56       & 1.1         & 1.1         & 1.94          & 103.92          & 0.17       & 0.17       & 0       & 0      & 2 \\
15 & 5 & 17.35        & 1.48        & 1.48        & 2.55          & 77.77           & 0.19       & 0.19       & 0       & 0      & 2 \\
25 & 3 & 94.87        & 1.49        & 1.49        & 2.53          & 4.04            & 0.18       & 0.18       & 3       & 0      & 2 \\\midrule
5  & 2 & -139.26      & 1.23        & 1.23        & 2.01          & 115.96          & 0.34       & 0.34       & 0       & 0      & 4 \\
10 & 3 & -47.79       & 1.46        & 1.46        & 2.42          & 181.01          & 0.14       & 0.14       & 0       & 0      & 4 \\
15 & 5 & 75.78        & 1.6         & 1.6         & 2.21          & 31.45           & 0.25       & 0.25       & 0       & 0      & 4 \\
25 & 3 & 99.33        & 1.71        & 1.71        & 1.93          &                 &            &            & 9       & 0      & 4 \\\midrule
5  & 2 & 37.88        & 1.8         & 1.8         & 1.75          & 43.37           & 0.79       & 0.79       & 0       & 0      & 8 \\
10 & 3 & 94.32        & 1.71        & 1.71        & 2.54          & 6.35            & 0.26       & 0.26       & 2       & 0      & 8 \\
15 & 5 & 99.23        & 1.78        & 1.78        & 3.34          & 0.21            & 0.19       & 0.19       & 7       & 0      & 8 \\
25 & 3 &              &             &             & 2.75          &                 &            &            & 10      & 0      & 8 \\ \bottomrule
\end{tabular}
\caption{RC201, default Gurobi settings.}
\label{tab:rc201_default}
\begin{tabular}{@{}llllllllllll@{}}
\toprule
N  & m & time\% & $\bar{\alpha}_{lb}$ & $\bar{\alpha}_{ub}$ & $\bar{t}_{SCM}$ & time\%std & $\sigma_{\alpha_{lb}}$ & $\sigma_{\alpha_{ub}}$ & TO & lb zero & T \\ \midrule
5  & 2 & -138.71      & 1.0         & 1.0         & 2.1           & 117.13          & 0.0        & 0.0        & 0       & 0      & 2 \\
10 & 3 & -54.29       & 1.1         & 1.1         & 2.05          & 73.67           & 0.21       & 0.21       & 0       & 0      & 2 \\
15 & 5 & -51.46       & 1.15        & 1.15        & 1.9           & 71.49           & 0.34       & 0.34       & 0       & 0      & 2 \\
25 & 3 & 96.35        & 1.52        & 1.52        & 2.7           & 3.48            & 0.38       & 0.38       & 1       & 0      & 2 \\\midrule
5  & 2 & -177.38      & 1.48        & 1.48        & 2.41          & 116.72          & 0.41       & 0.41       & 0       & 0      & 4 \\
10 & 3 & -110.61      & 1.18        & 1.18        & 2.39          & 173.02          & 0.2        & 0.2        & 0       & 0      & 4 \\
15 & 5 & -35.27       & 1.21        & 1.21        & 2.25          & 121.59          & 0.22       & 0.22       & 0       & 0      & 4 \\
25 & 3 & 94.57        & 1.36        & 1.36        & 2.99          & 2.87            & 0.13       & 0.13       & 1       & 0      & 4 \\\midrule
5  & 2 & -37.51       & 1.6         & 1.6         & 2.36          & 157.33          & 0.7        & 0.7        & 0       & 0      & 8 \\
10 & 3 & 84.05        & 1.45        & 1.45        & 2.36          & 11.97           & 0.44       & 0.44       & 0       & 0      & 8 \\
15 & 5 & 98.12        & 1.81        & 1.81        & 2.73          & 1.54            & 0.46       & 0.46       & 2       & 0      & 8 \\
25 & 3 &              &             &             & 2.86          &                 &            &            & 10      & 2      & 8 \\ \bottomrule
\end{tabular}
\caption{R101,  Gurobi settings MIPFocus=1, MIPGap=1.0.}
\label{tab:r101}
\begin{tabular}{@{}llllllllllll@{}}
\toprule
N  & m & time\% & $\bar{\alpha}_{lb}$ & $\bar{\alpha}_{ub}$ & $\bar{t}_{SCM}$ & time\%std & $\sigma_{\alpha_{lb}}$ & $\sigma_{\alpha_{ub}}$ & TO & lb zero & T \\ \midrule
5  & 2 & -243.59      & 1.4         & 0.75        & 2.62          & 164.42          & 0.33       & 0.23       & 0       & 0      & 2 \\
10 & 3 & -306.58      & 1.37        & 0.69        & 2.98          & 162.15          & 0.19       & 0.23       & 0       & 0      & 2 \\
15 & 5 & -169.79      & 1.56        & 0.51        & 2.99          & 134.75          & 0.28       & 0.16       & 0       & 0      & 2 \\
25 & 3 & -1.29        & 1.51        & 0.63        & 2.0           & 27.89           & 0.2        & 0.24       & 0       & 0      & 2 \\\midrule
5  & 2 & -85.5        & 2.33        & 0.83        & 1.72          & 92.5            & 1.05       & 0.22       & 0       & 0      & 4 \\
10 & 3 & -123.95      & 2.26        & 0.45        & 2.23          & 144.94          & 0.5        & 0.14       & 0       & 0      & 4 \\
15 & 5 & -102.51      & 2.12        & 0.66        & 2.98          & 116.46          & 0.58       & 0.2        & 0       & 0      & 4 \\
25 & 3 & 7.77         & 2.1         & 0.48        & 3.21          & 37.43           & 0.25       & 0.23       & 0       & 0      & 4 \\\midrule
5  & 2 & -313.57      & 2.55        & 0.67        & 2.82          & 295.07          & 1.44       & 0.27       & 0       & 0      & 8 \\
10 & 3 & -120.76      & 3.03        & 0.5         & 2.51          & 215.93          & 0.85       & 0.18       & 0       & 0      & 8 \\
15 & 5 & 39.88        & 3.13        & 0.45        & 2.85          & 39.89           & 0.69       & 0.16       & 0       & 0      & 8 \\
25 & 3 & 78.51        & 2.89        & 0.4         & 3.13          & 10.34           & 0.6        & 0.14       & 0       & 0      & 8 \\ \bottomrule
\end{tabular}
\end{table}
\begin{table}
\captionsetup{labelfont=bf,
              justification=raggedright,
              singlelinecheck=false}
\caption{R201, Gurobi settings MIPFocus=1, MIPGap=1.0}
\label{tab:r201}
\begin{tabular}{@{}llllllllllll@{}}
\toprule
N  & m & time\% & $\bar{\alpha}_{lb}$ & $\bar{\alpha}_{ub}$ & $\bar{t}_{SCM}$ & time\%std & $\sigma_{\alpha_{lb}}$ & $\sigma_{\alpha_{ub}}$ & TO & lb zero & T \\ \midrule
5  & 2 & -434.39      & 1.5         & 0.54        & 2.95          & 264.41          & 0.53       & 0.18       & 0       & 0      & 2 \\
10 & 3 &              &             &             & 3.82          &                 &            &            & 10      & 0      & 2 \\
15 & 5 &              &             &             & 2.6           &                 &            &            & 10      & 0      & 2 \\
25 & 3 & -273.25      & 1.85        & 0.36        & 4.77          & 217.48          & 0.24       & 0.05       & 0       & 0      & 2 \\\midrule
5  & 2 & -283.75      & 2.58        & 0.94        & 2.09          & 111.86          & 1.07       & 0.24       & 0       & 0      & 4 \\
10 & 3 & -413.32      & 2.24        & 0.49        & 3.72          & 345.66          & 1.04       & 0.22       & 1       & 0      & 4 \\
15 & 5 & -220.7       & 1.46        & 0.48        & 3.18          & 92.05           & 0.29       & 0.03       & 8       & 0      & 4 \\
25 & 3 & -197.27      & 2.82        & 0.29        & 5.14          & 88.92           & 1.56       & 0.1        & 0       & 0      & 4 \\\midrule
5  & 2 & -332.53      & 2.11        & 0.46        & 2.56          & 197.74          & 0.33       & 0.12       & 1       & 0      & 8 \\
10 & 3 & -179.63      & 3.17        & 0.38        & 2.55          & 125.31          & 0.41       & 0.22       & 4       & 0      & 8 \\
15 & 5 & -306.68      & 2.5         & 0.22        & 5.24          & 330.62          & 0.0        & 0.03       & 8       & 0      & 8 \\
25 & 3 & -71.94       & 4.5         & 0.13        & 6.75          & 53.23           & 0.76       & 0.03       & 2       & 2      & 8 \\ \bottomrule
\end{tabular}
\caption{C101, Gurobi settings MIPFocus=1, MIPGap=1.0}
\label{tab:c101}
\begin{tabular}{@{}llllllllllll@{}}
\toprule
 N  & m & time\% & $\bar{\alpha}_{lb}$ & $\bar{\alpha}_{ub}$ & $\bar{t}_{SCM}$ & time\%std & $\sigma_{\alpha_{lb}}$ & $\sigma_{\alpha_{ub}}$ & TO & lb zero & T \\ \midrule
 5  & 2 & -742.31      & 1.52        & 0.81        & 5.14          & 368.14          & 1.07       & 0.13       & 0       & 0      & 2 \\
 10 & 3 & -173.29      & 1.0         & 0.88        & 1.68          & 104.19          & 0.0        & 0.23       & 0       & 0      & 2 \\
 15 & 5 & -169.58      & 1.03        & 0.84        & 2.39          & 128.41          & 0.23       & 0.29       & 0       & 0      & 2 \\
 25 & 3 & -64.61       & 1.08        & 1.03        & 2.5           & 40.53           & 0.07       & 0.2        & 0       & 0      & 2 \\\midrule
 5  & 2 & -319.81      & 1.92        & 1.01        & 2.39          & 262.29          & 0.91       & 0.28       & 0       & 0      & 4 \\
 10 & 3 & -188.78      & 1.54        & 1.17        & 1.82          & 140.86          & 0.31       & 0.61       & 0       & 0      & 4 \\
 15 & 5 & -110.33      & 1.55        & 0.65        & 2.43          & 68.17           & 0.2        & 0.49       & 0       & 0      & 4 \\
 25 & 3 & -29.82       & 1.48        & 1.22        & 3.03          & 66.39           & 0.22       & 0.37       & 0       & 0      & 4 \\\midrule
 5  & 2 & -234.01      & 1.22        & 1.01        & 1.79          & 90.88           & 0.19       & 0.25       & 0       & 0      & 8 \\
 10 & 3 & -280.78      & 1.52        & 0.71        & 3.3           & 170.41          & 0.38       & 0.48       & 0       & 0      & 8 \\
 15 & 5 & 34.92        & 1.36        & 1.1         & 2.36          & 17.93           & 0.3        & 0.39       & 0       & 0      & 8 \\
 25 & 3 & 65.35        & 1.41        & 1.21        & 3.27          & 18.55           & 0.16       & 0.27       & 0       & 0      & 8 \\ \bottomrule
\end{tabular}
\caption{RC101, Gurobi settings MIPFocus=1, MIPGap=1.0}
\label{tab:rc101}
\begin{tabular}{@{}llllllllllll@{}}
\toprule
 N  & m & time\% & $\bar{\alpha}_{lb}$ & $\bar{\alpha}_{ub}$ & $\bar{t}_{SCM}$ & time\%std & $\sigma_{\alpha_{lb}}$ & $\sigma_{\alpha_{ub}}$ & TO & lb zero & T \\ \midrule
 5  & 2 & -162.78      & 1.8         & 0.49        & 2.16          & 123.63          & 0.79       & 0.06       & 0       & 0      & 2 \\
 10 & 3 & -181.42      & 1.63        & 0.52        & 1.94          & 62.17           & 0.54       & 0.22       & 0       & 0      & 2 \\
 15 & 5 & -156.89      & 1.68        & 0.55        & 2.55          & 87.75           & 0.29       & 0.18       & 0       & 0      & 2 \\
 25 & 3 & -22.9        & 2.02        & 0.48        & 2.58          & 45.34           & 0.35       & 0.12       & 0       & 0      & 2 \\\midrule
 5  & 2 & -164.28      & 2.95        & 0.74        & 2.01          & 108.86          & 0.76       & 0.17       & 0       & 0      & 4 \\
 10 & 3 & -257.51      & 2.69        & 0.51        & 2.42          & 171.1           & 1.01       & 0.14       & 0       & 0      & 4 \\
 15 & 5 & -67.93       & 2.41        & 0.52        & 2.21          & 44.88           & 0.67       & 0.18       & 0       & 0      & 4 \\
 25 & 3 & 13.35        & 3.46        & 0.46        & 2.88          & 44.39           & 1.21       & 0.17       & 0       & 0      & 4 \\\midrule
 5  & 2 & -114.31      & 2.4         & 0.73        & 1.75          & 132.17          & 0.7        & 0.56       & 0       & 0      & 8 \\
 10 & 3 & -122.4       & 3.45        & 0.38        & 2.63          & 107.53          & 1.26       & 0.09       & 0       & 0      & 8 \\
 15 & 5 & 12.4         & 3.2         & 0.18        & 2.94          & 29.6            & 0.42       & 0.03       & 0       & 0      & 8 \\
 25 & 3 & 73.95        & 4.07        & 0.28        & 2.75          & 7.49            & 1.14       & 0.12       & 0       & 0      & 8 \\ \bottomrule
\end{tabular}
\end{table}
\begin{table}
\captionsetup{labelfont=bf,
              justification=raggedright,
              singlelinecheck=false}
\caption{RC201, Gurobi settings MIPFocus=1, MIPGap=1.0.}
\label{tab:rc201}
\begin{tabular}{@{}llllllllllll@{}}
\toprule
N  & m & time\% & $\bar{\alpha}_{lb}$ & $\bar{\alpha}_{ub}$ & $\bar{t}_{SCM}$ & time\%std & $\sigma_{\alpha_{lb}}$ & $\sigma_{\alpha_{ub}}$ & TO & lb zero & T \\ \midrule
5  & 2 & -224.83      & 1.7         & 0.6         & 2.1           & 140.17          & 0.48       & 0.28       & 0       & 0      & 2 \\
10 & 3 &              &             &             & 2.05          &                 &            &            & 10      & 0      & 2 \\
15 & 5 &              &             &             & 1.9           &                 &            &            & 10      & 0      & 2 \\
25 & 3 & -12.26       & 1.95        & 0.38        & 2.64          & 46.93           & 0.16       & 0.08       & 0       & 0      & 2 \\\midrule
5  & 2 & -251.14      & 3.1         & 0.77        & 2.41          & 175.49          & 0.99       & 0.26       & 0       & 0      & 4 \\
10 & 3 &              &             &             & 2.39          &                 &            &            & 10      & 0      & 4 \\
15 & 5 &              &             &             & 2.25          &                 &            &            & 10      & 0      & 4 \\
25 & 3 & -20.35       & 3.88        & 0.34        & 2.92          & 50.16           & 1.89       & 0.1        & 0       & 0      & 4 \\\midrule
5  & 2 & -377.52      & 2.67        & 0.69        & 2.36          & 269.57          & 1.15       & 0.56       & 7       & 0      & 8 \\
10 & 3 &              &             &             & 2.36          &                 &            &            & 10      & 0      & 8 \\
15 & 5 &              &             &             & 2.73          &                 &            &            & 10      & 0      & 8 \\
25 & 3 & 43.65        & 4.25        & 0.21        & 2.98          & 22.35           & 0.71       & 0.24       & 2       & 2      & 8 \\ \bottomrule
\end{tabular}
\end{table}
\section{Conclusions and Future Work}
This paper presents a heuristic algorithm, SCM, for  solving the demand robust fleet size problem with time windows and compatibility constraints.  The proposed SCM heuristic performs comparable, or better than Gurobi while the reduction in computation time scales exponentially with the size of the problem, in both theory and simulation.  Natural directions for future work will involve determining the approximation ratio for SCM, and also comparing the heuristic with approaches from constraint programming.  Work may also be done to understand the dependence of SCM's performance on the data, especially for clustered instances.  These future directions will provide insight into a broader class of large-scale practical problems for which mappings to set cover perform remarkably well.
\bibliographystyle{splncs04}
\bibliography{refs.bib}
%
%
%
%
%
\end{document}